\RequirePackage{etoolbox}
\csdef{input@path}{%
 {sty/}
 {img/}
}%
\csgdef{bibdir}{bib/}

\documentclass[ba]{imsart}
\pubyear{0000}
\volume{00}
\issue{0}
\doi{0000}
\firstpage{1}
\lastpage{1}

\usepackage{amsthm}
\usepackage{amsmath}
\usepackage{natbib}
\usepackage[colorlinks,citecolor=blue,urlcolor=blue,filecolor=blue,backref=page]{hyperref}
\usepackage{graphicx}

\usepackage{amsmath,amsfonts,amsthm}
\usepackage{graphicx,psfrag,epsf}
\usepackage{enumerate}
\usepackage{natbib}
\usepackage{url} 
\usepackage{algorithm}
\usepackage{algorithmic}

\newtheorem{theorem}{Theorem}
\newtheorem{proposition}{Proposition}
\mathchardef\mhyphen="2D

\startlocaldefs
\numberwithin{equation}{section}
\theoremstyle{plain}

\endlocaldefs

\begin{document}

\begin{frontmatter}
\title{Variational Hamiltonian Monte Carlo via Score Matching}
\runtitle{Variational Hamiltonian Monte Carlo}

\begin{aug}
\author{\fnms{Cheng} \snm{Zhang}\thanksref{addr1}\ead[label=e1]{chengz4@uci.edu}},
\author{\fnms{Babak} \snm{Shahbaba}\thanksref{addr2}\ead[label=e2]{babaks@uci.edu}}
\and
\author{\fnms{Hongkai} \snm{Zhao}\thanksref{addr3}
\ead[label=e3]{zhao@math.uci.edu}}

\runauthor{Zhang et al.}

\address[addr1]{UC Irvine, 
    \printead{e1}}

\address[addr2]{UC Irvine, 
    \printead{e2}}

\address[addr3]{UC Irvine, 
    \printead{e3}}

\end{aug}

\begin{abstract}
Traditionally, the field of computational Bayesian statistics has been divided into two main subfields: variational methods and Markov chain Monte Carlo (MCMC). In recent years, however, several methods have been proposed based on combining variational Bayesian inference and MCMC simulation in order to improve their overall accuracy and computational efficiency. This marriage of fast evaluation and flexible approximation provides a promising means of designing scalable Bayesian inference methods. In this paper, we explore the possibility of incorporating variational approximation into a state-of-the-art MCMC method, Hamiltonian Monte Carlo (HMC), to reduce the required expensive computation involved in the sampling procedure, which is the bottleneck for many applications of HMC in big data problems. To this end, we exploit the regularity in parameter space to construct a free-form approximation of the target distribution by a fast and flexible surrogate function using an optimized additive model of proper random basis, which can be viewed as a single-hidden layer feedforward neural networks as well. The surrogate provides sufficiently accurate approximation while allowing for fast computation in the sampling procedure, resulting in an efficient approximate inference algorithm. We demonstrate the advantages of our method on both synthetic and real data problems.
\end{abstract}



\end{frontmatter}

\section{Introduction}

Bayesian inference has been successful in modern data analysis. Given a probabilistic model for the underlying mechanism of the observed data, Bayesian methods properly quantify uncertainty and reveal the landscape or global structure of parameter space. While conceptually simple, exact posterior inference in many Bayesian models is often intractable. Therefore, in practice, people often resort to approximation methods among which Markov chain Monte Carlo (MCMC) and variational Bayes (VB) are the two most popular choices. 

The MCMC approach is based on drawing a series of correlated samples by constructing a Markov chain with guaranteed convergence to the target distribution. Therefore, MCMC methods are asymptotically unbiased. Simple methods such as random-walk Metropolis \citep{metropolis53}, however, often suffer from slow mixing (due to their random walk nature) when encountering complicated models with strong dependencies among parameters. Introducing an auxiliary momentum variable, Hamiltonian Monte Carlo (HMC) \citep{duane87,neal11} reduces the random walk behavior by proposing states following a Hamiltonian flow which preserves the target distribution. By incorporating the geometric information of the target distribution, e.g., the gradient, HMC is able to generate distant proposals with high acceptance probabilities, enabling more efficient exploration of the parameter space than standard random-walk proposals. 

A major bottleneck of HMC, however, is the computation of the gradient of the potential energy function in order to simulate the Hamiltonian flow. As the datasets involved in many practical tasks, such as ``big data'' problems, usually have millions to billions of observations, such gradient computations are infeasible since they need full scans of the entire dataset. In recent years, many attempts have been made to develop scalable MCMC algorithms that can cope with very large data sets \citep{Welling11,Ahn12,chen14,Ding14}. The key idea of these methods stems from stochastic optimization where noisy estimates of the gradient based on small subsets of the data are utilized to scale up the algorithms.  The noise introduced by subsampling, however, could lead to non-ignorable loss of accuracy, which in turn hinders the exploration efficiency of standard MCMC approaches \citep{betancourt15}.

The main alternative to MCMC is variational Bayes inference \citep{jordan99,wainwright08}. As a deterministic approach, VB transforms Bayesian inference into an optimization problem where a parametrized distribution is introduced to fit the target posterior distribution by minimizing the Kullback-Leibler (KL) divergence with respect to the variational parameters. Compared to MCMC methods, VB introduces bias but is usually faster. 

A natural question would be: can we combine both methods to mitigate the drawbacks and get the best of both worlds? The first attempt in this direction was proposed by \cite{Freitas01} where a variational approximation was used as proposal distribution in a block Metropolis-Hasting (MH) MCMC kernel to locate the high probability regions quickly, thus facilitating convergence. Recently, a new synthesis of variational inference and Markov chain Monte Carlo methods has been explored in \cite{Salimans15} where one or more steps of MCMC are integrated into variational approximation. The extra flexibility from MCMC steps provides a rich class of distributions to find a closer fit to the exact posterior, which allows for further improvement on the approximation quality. 


In this work, we explore the possibility of utilizing variational approximation to speed up HMC for problems with large scale datasets. The key idea is to integrate fast variational approximations into the sampling procedure so that the overall computational complexity can be reduced. The main contributions are summarized as follows:
\begin{enumerate}
\item We train a fast neural network surrogate using extreme learning machine (ELM) to approximate the log-posterior (potential energy function). Neural networks with randomly-mapped features (with almost any nonlinear activation functions) has been proved to approximate a rich class of functions arbitrarily well in \cite{huang06chen}. After the fast training of a reasonably accurate neural network surrogate, we use its gradient to simulate the surrogate induced Hamiltonian flow in order to generate proposals. The learning and training of the surrogate function provides an effective and flexible way to explore both regularity of the parameter space and redundancy in the data collectively, which can be viewed as an implicit subsampling. However, unlike the popular subsampling-based methods where the \emph{leapforg} stepsize needs to be annealed or remains small enough to compensate for the resulting bias, the \emph{leapfrog} stepsize in our approach can remain close to that of standard HMC while keeping a comparable acceptance probability. As such, the exploration efficiency of HMC can be maintained while reducing the computational cost.

\item A new training procedure is proposed for an efficient and general surrogate method. The neural network surrogate is trained by minimizing the squared distance between the gradient of the surrogate and the gradient of the target (log-posterior), a procedure that resembles \emph{score matching} \citep{hyvarinen05}. The training data are collected while the ``modified'' HMC sampler (based on the surrogate induced Hamiltonian flow) explores the parameter space. In \cite{hyvarinen05}, the training data are assumed to be sampled from the target distribution. In our method, by including the gradient of the target distribution in the learning process, the information from the target distribution is explicitly included in the surrogate in our method. Therefore, the restriction on the training data is relaxed.

\item To further reduce the computation cost, we also use the computationally fast surrogate in the Metropolis-Hastings correction step. As a result, the modified HMC sampler will no longer converge to the exact target distribution. Instead, it will converge to the best approximation from an exponential family with pre-specified random sufficient statistics (which we call \emph{free-form} variational approximation in the sequel). This variational perspective distinguishes our approach from the existing surrogate methods on accelerating HMC. Moreover, compared to traditional {\it fixed-form} variational approximations, the \emph{free-form} variational approximation is usually more flexible and thus can provide a better fit to more general target distributions.
\end{enumerate}

Our paper is organized as follows. In section \ref{sec:background}, we introduce the two ingredients related to our method: Hamiltonian Monte Carlo and {\it fixed-form} variational Bayes. Section \ref{sec:vhmc} presents our method, termed Variational Hamiltonian Monte Carlo (VHMC). We demonstrate the efficiency of VHMC in a number of experiments in section \ref{sec:examples} and conclude in section \ref{sec:conc}.

\section{Background} \label{sec:background}
\subsection{Hamiltonian Monte Carlo}
In general formulation of Bayesian inference, a set of independent observations $Y = \{y_1,\ldots,y_N\}$ are modeled by an underlying distribution $p(y|\theta)$ with unknown parameter $\theta$. Given a prior distribution of $\theta \sim p(\theta)$, the posterior distribution is given by Bayesian formula
\begin{equation*}
p(\theta|Y) = \frac{p(Y|\theta)p(\theta)}{p(Y)} \propto \prod_{n=1}^N\;p(y_n|\theta)\cdot p(\theta)
\end{equation*}
To construct the Hamiltonian dynamical system \citep{duane87,neal11}, the position-dependent potential energy function is defined as the negative log unnormalized posterior density
\begin{equation}\label{eq:potential}
U(\theta) = -\sum_{n=1}^N\;\log p(y_n|\theta) - \log p(\theta)
\end{equation}
and the kinetic energy function is defined as a quadratic function of an auxiliary momentum variable $r$
\begin{equation*}
K(r) = \frac12r^TM^{-1}r
\end{equation*}
where $M$ is a mass matrix and is often set to identity, $I$. The fictitious Hamiltonian, therefore, is defined as the total energy function of the system
\begin{equation*}
H(\theta,r) = U(\theta) + K(r)
\end{equation*}
As one of the state-of-the-art MCMC methods, Hamiltonian Monte Carlo suppresses random walk behavior by simulating the Hamiltonian dynamical system to propose distant states with high acceptance probabilities. That is, in order to sample from the posterior distribution $p(\theta|Y)$, HMC augments the parameter space and generates samples from the joint distribution of $(\theta,r)$
\begin{equation}\label{eq:joint}
\pi(\theta,r) \propto \exp(-U(\theta)-K(r))
\end{equation}
Notice that $\theta$ and $r$ are separated in \eqref{eq:joint}, we can simply drop the momentum samples $r$ and the $\theta$ samples follow the marginal distribution which is exactly the target posterior.

To generate proposals, HMC simulates the Hamiltonian flow governed by the following differential equations
\begin{align}
\frac{d\theta}{dt} &= \frac{\partial H}{\partial r} = M^{-1}r \label{eq:velocity}\\
\frac{dr}{dt} &= - \frac{\partial H}{\partial \theta} = -\nabla_{\theta} U(\theta) \label{eq:force}
\end{align}
Over a period $t$, also called trajectory length, \eqref{eq:velocity} and \eqref{eq:force} together define a map $\phi_t: (\theta_0,r_0) \mapsto (\theta^\ast,r^\ast)$ in the extended parameter space, from the starting state to the end state. As implied by a Hamiltonian flow, $\phi_t$ is reversible, volume-preserving and also preserves the Hamiltonian $H(\theta_0,r_0) = H(\theta^\ast,r^\ast)$. These allow us to construct $\pi$-invariant Markov chains whose proposals will always be accepted. In practice, however, \eqref{eq:velocity} and \eqref{eq:force} are not analytically solvable and we need to resort to numerical integrators. As a symplectic integrator, the {\it leapfrog} scheme (see Algorithm \ref{alg:hmc}) maintains reversibility and volume preservation and hence is a common practice in HMC literatures. Although discretization introduces bias which needs to be corrected in an Metroplis-Hasting (MH) step, we can control the stepsizes to maintain high acceptance probabilities even for distant proposals.

\begin{algorithm}[!t]
\begin{algorithmic}
\STATE{\bfseries Input:} Starting position $\theta^{(1)}$ and step size $\epsilon$
\FOR{$t =1,2,\ldots,T$}
  \STATE \textit{Resample momentum $r$}
  \STATE $r^{(t)} \sim \mathcal{N}(0,M)$
  \STATE $(\theta_0,r_0) = (\theta^{(t)},r^{(t)})$
  \STATE \textit{Simulate discretization of Hamiltonian dynamics:}
  \FOR{$l = 1$ {\bfseries to} $L$} 
  \STATE $r_{l-1} \leftarrow r_{l-1} - \frac{\epsilon}{2} \nabla_\theta U(\theta_{l-1})$
  \STATE $\theta_l \leftarrow \theta_{l-1} + \epsilon M^{-1}r_{l-1}$
  \STATE $r_l \leftarrow r_l - \frac{\epsilon}{2} \nabla_\theta U(\theta_{l})$
  \ENDFOR
  \STATE $(\theta^{\ast},r^{\ast}) = (\theta_L,r_L)$
  \STATE \textit{Metropolis-Hasting correction:} 
  \STATE $u \sim \text{Uniform}[0,1]$
  \STATE $\rho = \exp[{ H(\theta^{(t)},r^{(t)})-H(\theta^{\ast},r^{\ast}) }]$
  \IF{$u < \min(1,\rho)$} 
  \STATE $\theta^{(t+1)} = \theta^{\ast}$ 
  \ELSE
  \STATE $\theta^{(t+1)} = \theta^{(t)}$
  \ENDIF
  \ENDFOR
 \end{algorithmic}
\caption{Hamiltonian Monte Carlo}
\label{alg:hmc}
\end{algorithm}

In recent years, many variants of HMC have been developed to make the algorithm more flexible and generally applicable in a variety of settings. For example, methods proposed in \cite{hoffman11} and \cite{wang13} enable automatically tuning of hyper-paramters such as the stepsize $\epsilon$ and the number of {\it leapfrog} steps $L$, saving the amount of tuning-related headaches. Riemannian Manifold HMC \citep{RMHMC} further improves standard HMC's efficiency by automatically adapting to local structures using Riemanian geometry of parameter space. These adaptive techniques could be potentially combined with our proposed method which focuses on reducing the computational complexity.

\subsection{Fixed-form Variational Bayes}
Instead of running a Markov chain, people can also approximate the intractable posterior distribution with a more convenient and tractable distribution. A popular approach of obtaining such an approximation is {\it fixed-form variational Bayes} \citep{Honkela10,saul96,salimans13} where a parametrized distribution $q_\eta(\theta)$ is proposed to approximate the target posterior $p(\theta|Y)$ by minimizing the KL divergence 
\begin{align}\label{eq:kl}
D_{KL}&(q_\eta(\theta)||p(\theta|Y)) = \int q_\eta(\theta)\log\left(\frac{q_\eta(\theta)}{p(\theta|Y)}\right) d\theta\nonumber\\
&= \log(p(Y)) + \int q_\eta(\theta)\log\left(\frac{q_\eta(\theta)}{p(\theta,Y)}\right) d\theta
\end{align}
since $\log(p(Y))$ is a constant (used extensively in model selection), it suffices to minimize the second term in \eqref{eq:kl}. Usually, $q_\eta(\theta)$ is chosen from the exponential family of distributions with the following canonical form:
\begin{equation}\label{eq:approx}
q_\eta(\theta) = \exp[T(\theta)\eta-A(\eta)]\nu(\theta)
\end{equation}
where $T(\theta)$ is a row vector of sufficient statistics, $A(\eta)$ is for normalization and $\nu(\theta)$ is a base measure. The column vector $\eta$ is often called the natural parameters of the exponential family distribution $q_\eta(\theta)$. Taking this approach and substituting into \eqref{eq:kl}, we now have a parametric optimization problem in $\eta$:
\begin{equation}\label{eq:op}
\hat{\eta} = \arg\min_{\eta} \mathbb{E}_{q_\eta(\theta)}[\log q_{\eta}(\theta) - \log p(\theta,Y)]
\end{equation}
The above optimization problem can be solved using gradient-based optimization or fix-point algorithms if $\mathbb{E}_{q_\eta(\theta)}[\log q_\eta(\theta)],\;\mathbb{E}_{q_\eta(\theta)}[\log p(\theta,Y)]$ and its derivatives with respect to $\eta$ can be evaluated analytically. Without assuming posterior independence and requiring conjugate exponential models, posterior approximations of this type are usually much more accurate than a factorized approximation following the mean-field assumptions. However, the requirement of being able to analytically evaluate those quantities mentioned above is also very restrictive. Many attempts, therefore, has been made to evaluate those quantities using stochastic approximations from Monte Carlo integration \citep{blei12,RGB13,Kingma13}. As an alternative, \cite{salimans13} proposed a new optimization algorithm which relates \eqref{eq:op} to stochastic linear regression. To reveal the connection, the posterior approximate \eqref{eq:approx} is relaxed and rewritten in the unnormalized form 
\begin{equation}\label{eq:adjust}
\tilde{q}_{\tilde{\eta}}(\theta) = \exp[\tilde{T}(\theta)\tilde{\eta}]\nu(\theta)
\end{equation}
where the nonlinear normalizer $A(\eta)$ is removed and the vectors of sufficient statistics and natural parameters are augmented, i.e. $\tilde{T}(\theta) = (1,T(\theta)),\; \tilde{\eta} = (\eta_0,\eta')'$. The unnormalized version of KL divergence is utilized to deal with $\tilde{q}_{\tilde{\eta}}(\theta)$ and achieves its minimum at  
\begin{equation}\label{eq:lr}
\tilde{\eta} = \mathbb{E}_{q}[\tilde{T}(\theta)'\tilde{T}(\theta)]^{-1}\mathbb{E}_{q}[\tilde{T}(\theta)'\log p(\theta,Y)]
\end{equation}
which resembles the maximum likelihood estimator for linear regression. Based on this observation, \cite{salimans13} derived a stochastic approximation algorithm using \eqref{eq:lr} as a fixed point update and approximating the involved expectations by weighted Monte Carlo. 

In the next section, we will discuss how the variational Bayes approach can be actually utilized to accelerate HMC. For this, we construct a fast and accurate approximation for the computationally expensive potential energy function. The approximation is provided by variational Bayes and is incorporated in the simulation of Hamiltonian flow.


\section{Variational Hamiltonian Monte Carlo}\label{sec:vhmc}
Besides subsampling, an alternative approach that can save computation cost is to construct fast and accurate surrogate functions for the expensive potential energy functions \citep{liu01,neal11}. As one of the commonly used models for emulating expensive-to-evaluate functions, Gaussian process (GP) is used in \cite{rasmussen03} to approximate the potential energy and its derivatives based on true values of these quantities (training set) collected during an initial exploratory phase. However, a major drawback of GP-based surrogate methods is that inference time typically grows cubically in the size of training set due to the necessity of inverting a dense covariance matrix. This is especially crucial in high dimensional spaces, where large training sets are often needed before a reasonable level of approximation accuracy is achieved. Our goal, therefore, is to develop a method that can scale to large training set while still maintaining a  desired level of flexibility. For this purpose, we propose to use neural networks along with efficient training algorithms to construct surrogate functions. A typical single-hidden layer feedforward neural network (SLFN) with scalar output is defined as
\begin{equation}\label{eq:slfn}
z(\theta) = \sum_{i=1}^sv_i\sigma(\theta;\gamma_i)
\end{equation}
where $\gamma_i$ and $v_i$ are the node parameter and output weight for the {\it i}th hidden neuron, $\sigma$ is a nonlinear activation function. Given a training dataset,
the estimates of those parameters can be obtained by minimizing the mean square error (MSE) cost function. 
To save training time, randomly assigned node parameters $\{\gamma_i\}_{i=1}^s$ are suggested in \cite{ferrari05} and \cite{huang06} where the optimization is reduced to a linear regression problem with randomly mapped features which can be solved efficiently using algebraic approaches. This can also be viewed as using an additive model based on random (adaptive) basis to approximate the target distribution. Unlike a standard Gaussian process, the above neural network based surrogate scales linearly in the size of training data, and cubically in the number of hidden neurons. This allows us to explicitly balance evaluation time and model capacity. Moreover, the random network in \eqref{eq:slfn} has been proven to approximate a rich class of functions arbitrarily well \citep{rahimi08}.

\subsection{Approximation with random networks}
As shown in \cite{rahimi08}, functions in \eqref{eq:slfn} with randomly assigned bases is flexible enough to provide accurate approximations for other well-studied classes of functions (e.g., Reproducing Kernel Hilbert Space). Let $\{\sigma(\cdot;\gamma): \gamma \in \Gamma\}$ be a family of functions on a compact set $\Theta\subset \mathbb{R}^d$ with parameter $\gamma$ specified over the set $\Gamma$. Let $p$ be a distribution on  $\Gamma$, consider a rich class of functions of the following form
\begin{equation}\label{eq:mixact}
 f(\theta) = \int_{\Gamma} \alpha(\gamma)\sigma(\theta;\gamma)\;d\gamma
\end{equation}
where $|\alpha(\gamma)| \leq C |p(\gamma)|,\;\forall \gamma \in \Gamma$ for some constant $C$. Define a norm $\|f\|_p = \sup_{\gamma} \left|\frac{\alpha(\gamma)}{p(\gamma)}\right|$ and the set
\begin{equation}
\mathcal{F}_p\equiv \left\{f(\theta)=\int_{\Gamma}\alpha(\gamma)\sigma(\theta;\gamma)\;d\gamma\;\big|\;\|f\|_p < \infty \right\}
\end{equation}
The following theorem shows that a given $f\in\mathcal{F}_p$ can be approximated within $\mathcal{O}(\|f\|_p/\sqrt{s})$ by a function of the form 
\begin{equation}
z(\theta) = \sum_{i=1}^sv_i\sigma(\theta;\gamma_i)
\end{equation}
where $\gamma_1,\ldots,\gamma_s$ are sampled iid from $p(\gamma)$. See \citet{rahimi08} for a detailed proof.

\begin{theorem}[Rahimi 2008]
Let $\mu$ be any probability measure on $\Theta$, and define the norm $\|f\|^2_\mu=\int_{\Theta}f^2(\theta)\mu(d\theta)$. Suppose $\sigma$ satisfies $\sup_{\theta,\gamma}|\sigma(\theta;\gamma)|\leq 1$. Fix $f\in\mathcal{F}_p$. Then, $\forall \delta >0$, with probability at least $1-\delta$ over $\gamma_1,\ldots,\gamma_s$ drawn iid from $p$, there exist $v_1,\ldots,v_s$ so that the function 
\[
z(\theta) = \sum_{i=1}^sv_i\sigma(\theta;\gamma_i)
\]
satisfies 
\begin{equation}
\|z-f\|_\mu < \frac{\|
f\|_p}{\sqrt{s}}\left(1+\sqrt{2\log\frac{1}{\delta}}\right)
\end{equation}
\end{theorem}

The activation function $\sigma$ and probability distribution $p(\gamma)$ on $\Gamma$ together defines a Reproducing Kernel Hilbert Space (RKHS) $\mathcal{H}$ with the following kernel $k$ on $\Theta\times\Theta$
\begin{equation}\label{eq:kernelRKHS}
k(\theta,\theta') = \int_{\Gamma} p(\gamma)\sigma(\theta;\gamma)\sigma(\theta';\gamma)d\gamma
\end{equation}
which is clearly positive definite. Alternatively, the following proposition shows that $\mathcal{H}$ can be constructed based on functions of the form \eqref{eq:mixact}. 

\begin{proposition}
Let the space $\hat{\mathcal{H}}$ be the completion of the set of functions of the form \eqref{eq:mixact} such that
\[
\int_{\Gamma} \frac{\alpha(\gamma)^2}{p(\gamma)}\;d\gamma < \infty
\]
with the inner product 
\[
\langle f, g\rangle = \int_{\Gamma} \frac{\alpha(\gamma)\beta(\gamma)}{p(\gamma)}\;d\gamma
\]
where $g(\theta) = \int_{\Gamma} \beta(\gamma)\sigma(\theta;\gamma)\;d\gamma$. Then $\hat{\mathcal{H}} = \mathcal{H}$.
\end{proposition}
\begin{proof}
See a proof in \citet{rahimi08}.
\end{proof}
Note that $\forall f\in \mathcal{F}_p$, 
\[
\int_{\Gamma} \frac{\alpha(\gamma)^2}{p(\gamma)} = \int_{\Gamma}\frac{\alpha(\gamma)^2}{p(\gamma)^2}p(\gamma)d\gamma \leq \|f\|^2_p <\infty
\]
Therefore, $\mathcal{F}_p$ is a subset of RKHS $\mathcal{H}$. In fact, \citet{rahimi08} shows that $\mathcal{F}_p$ is a dense subset of $\mathcal{H}$.
\begin{theorem}
Let $\mathcal{F}_p$ and $\mathcal{H}$ be defined as above for a given activation function $\sigma(\theta;\gamma)$ and probability density $p(\gamma)$. Then $\mathcal{F}_p$ is dense in $\mathcal{H}$.
\end{theorem}
\begin{proof}
By the property of RKHS, $\mathcal{H}$ is the completion of the set of all finite linear combinations of the form
\begin{equation}\label{eq:spankernel}
f(\theta) = \sum_i v_ik(\theta;\theta_i),\quad \theta_i \in \Theta
\end{equation}
with the inner product satisfying $\langle k(\cdot,\theta_i),k(\cdot,\theta_j)\rangle = k(\theta_i,\theta_j)$. Using \eqref{eq:kernelRKHS}, we can rewrite \eqref{eq:spankernel} in the following form 
\[
f(\theta) = \sum_iv_i\int_{\Gamma} p(\gamma)\sigma(\theta;\gamma)\sigma(\theta_i;\gamma)\;d\gamma = \int_{\Gamma} p(\gamma)\sum_i v_i\sigma(\theta_i;\gamma) \sigma(\theta;\gamma) \; d\gamma = \int_{\Gamma} \alpha_f(\gamma)\sigma(\theta;\gamma) \; d\gamma
\]
where $\alpha_f(\gamma) = p(\gamma)\sum_i v_i\sigma(\theta_i;\gamma) \Rightarrow \left|\frac{\alpha_f(\gamma)}{p(\gamma)}\right| \leq \sum_i|v_i|, \forall \gamma\in {\Gamma}$, implying that $f\in\mathcal{F}_p$.
\end{proof}
Since certain RKHSs are known to fit a rich class of (density) functions arbitrarily well, approximating these RKHSs with a small number of random bases allows for efficient surrogate construction with comparable approximation accuracy.

\subsection{Free-form Variational Bayes}\label{sec:ffvb}
The correspondence between distributions and their potential energy functions builds a bridge between distribution approximation and function approximation. Viewing this way, each random neural network in \eqref{eq:slfn} corresponds to a distribution in the exponential family 
\begin{equation}\label{eq:free}
q_{v}(\theta) \propto  \exp(-z(\theta)) = \exp[-\sum_{i=1}^sv_i\sigma(\theta;\gamma_i)-\Phi(v)]
\end{equation}
where $v=(v_i,\;i=1,2,\ldots,s)$ is called the vector of canonical parameters, and the collection of randomly-mapped features $\Psi = (\Psi_i,\;i=1,2,\ldots,s),\;\Psi_i=-\sigma(\theta;\gamma_i),\;i=1,2,\ldots,s$ is known as sufficient statistics. The quantity $\Phi$, known as the log partition function, is defined by the following integral
\[
\Phi(v) = \log \int \exp(v\cdot \Psi(\theta))\;d\theta
\]
Note that $q_v(\theta)$ is properly normalized if and only if the integral is finite. Therefore, the canonical parameters $v$ of interest belong to the set 
\begin{equation}\label{eq:paradom}
\Omega := \{v\in\mathbb{R}^s|\Phi(v)<+\infty\}
\end{equation}
We call $q_v(\theta)$ the induced distribution of neural network surrogate $z(\theta)$ if $v\in\Omega$.

As our neural network surrogates approximate the true potential energy function $U(\theta)$, the underlying distribution $q_{v}(\theta)$ then approximates the target posterior distribution $p(\theta|Y)$. Because both the surrogate induced distribution $q_v(\theta)$ and the target posterior distribution $p(\theta|Y)$ are known up to a constant, we introduce the following expected squared distance between unnormalized densities 
\begin{equation}\label{eq:smdist}
\tilde{D}_{SM}(p(\theta|Y)||q_v(\theta)) = \frac12\int q_v(\theta)\| \nabla_\theta z(\theta)-\nabla_\theta U(\theta)\|^2\;d\theta
\end{equation}
which is similar to the well known score matching squared distance \citep{hyvarinen05} (see section \ref{sec:score} for a more detailed explanation). By minimizing the expected squared distance $\tilde{D}_{SM}$, we arrive at a variational inference algorithm
\begin{equation}\label{eq:freevar}
\hat{v} = \arg\min_{v\in\Omega}\tilde{D}_{SM}(p(\theta|Y)||q_v(\theta))
\end{equation}

The surrogate induced approximation \eqref{eq:free} enriches our choices for variational approximation from {\it fixed-form} tractable distributions (e.g., Gaussian or mixture of Gaussians) to fast and flexible intractable distributions. The integral in \eqref{eq:smdist} is usually hard to evaluate in practice but can be approximated using samples from the surrogate induced distribution $q_v(\theta)$ by the law of large numbers. Unlike the {\it fixed-form} approximation, the surrogate induced distribution generally does not allow for drawing samples directly. However, we can use MCMC methods to generate samples from it. 

Due to the random nature in approximation \eqref{eq:free}, we call variational algorithm \eqref{eq:freevar} {\it free-form} variational Bayes. By choosing a proper number of hidden neurons $s$, the {\it free-form} variational Bayes provides an implicit subsampling procedure that can effectively remove redundancy and noise in the data while striking a good balance between computation cost and approximation accuracy of the underlying distribution.

{\bf Remark.}
From an approximation point of view, each $\sigma(\theta;\gamma)$ is a basis with a random configuration of $\gamma$, which specifies the orientation and scaling properties, within a profile of the activation function $\sigma$.  In particular we choose the  \emph{softplus} function for additive node $\sigma(\theta;w,d) = \log\big(1+\exp(w\cdot\theta + d)\big)$ to approximate the potential function, e.g., the Hamiltonian corresponding to the posterior distribution, for our free-form variational Bayes. Of course, the choice of $\sigma$ is general. As shown in \cite{NNSHMC}, different basis can be used in our free-form variational Bayes formulation. In particular, if exponential square kernel is used as radial basis functions (RBF) centered at given data, GP method is recovered. However, using kernel functions centered at very data point, GP method does not effectively exploit redundancy in data or regularity in parameter space and hence may result in expensive computation cost as well as instability for large data set. If exponential square kernel is used at a smaller set of properly induced points, a more computationally tractable GP model that tries to exploit redundancy in data, sparse GP \citep{snelson06,joaquin05}, is recovered. Since kernel function is usually local, to approximate potential functions in HMC which go to infinity at far field, it is shown in  \cite{NNSHMC} that the use of random basis based on softplus function provides a more compact and better approximation that has a good balance between local features and global behaviors. Also the choice of $s$, the number of basis, can be used to balance flexibility/accuracy and overfitting.

\subsection{Surrogate Induced Hamiltonian Flow}
To sample from the surrogate induced distribution $q_v(\theta)$, we generate proposals by simulating the corresponding surrogate induced Hamiltonian flow governed by the following equations
\begin{align}
\frac{d\theta}{dt} &= \frac{\partial \tilde{H}}{\partial r} = M^{-1}r \label{eq:survel}\\
\frac{dr}{dt} &= - \frac{\partial \tilde{H}}{\partial \theta} = -\nabla_{\theta} z(\theta) \label{eq:surforce}
\end{align}
where the modified Hamiltonian is 
\[
\tilde{H}(\theta,r) = z(\theta) + K(r)
\]
Practitioners can use \emph{leapfrog} scheme to solve \eqref{eq:survel} and \eqref{eq:surforce} numerically. The end-point $(\theta^\ast,r^\ast)$ of the trajectory starting at $(\theta_0,r_0)$ is accepted with probability
\begin{equation}\label{eq:apsur}
\alpha_{\mathrm{vhmc}} = \min[1,\exp(\tilde{H}(\theta_0,r_0)-\tilde{H}(\theta^\ast,r^\ast))]
\end{equation}
Similar to the true Hamiltonian flow, we have the following theorem
\begin{theorem}\label{thm:surHL}
If we construct a Markov chain by simulating surrogate induced Hamiltonian flow \eqref{eq:survel} and \eqref{eq:surforce} using \emph{leapfrog} steps and Metropolis-Hastings correction with the acceptance probability according to \eqref{eq:apsur}, the equilibrium distribution of the chain is 
\[
\tilde{\pi}(\theta,r)\propto \exp(-z(\theta)-K(r))
\]
\end{theorem}
See Supplementary Materials for a proof.

Theorem \ref{thm:surHL} implies that the marginal distribution of $\theta$ is exactly the surrogate induced distribution $q_v(\theta)$. Therefore, we can run the Markov chain and collect the values of interest, such as the potential energy function and its derivatives, as additional training data to improve the surrogate approximation (solving \eqref{eq:freevar}) on the fly. This way, our algorithm can be viewed as a generalized version of the stochastic approximation algorithm proposed in \cite{salimans13} for \emph{free-form} variational Bayes.

\subsection{Score Matching}\label{sec:score}
A well-known strategy to estimate unnormalized models is score matching \citep{hyvarinen05}. Assuming that data $D$ come from a distribution with unknown density $p_{D}(.)$, we want to find an approximation with a parameterized unnormalized density model $q_v(.)$, where $v$ is an $s$-dimensional vector of parameters. Score matching estimates the model by minimizing the expected squared distance between the model score function $\psi_v(\theta) = \nabla_\theta\log q_v(\theta)$ and the data score function $\psi_{D}(\theta) = \nabla_\theta\log p_{D}(\theta)$
\begin{equation}
J(v) = \frac12\int p_D(\theta)\|\psi_v(\theta)-\psi_D(\theta)\|^2\;d\theta
\end{equation}
A simple trick was suggested in \cite{hyvarinen05} to avoid the expensive estimation of the data score function $\psi_D(\theta)$ from $D$.

Similar ideas can be applied here to train our neural network surrogate. Notice that the posterior density is known up to a constant, the data score function then can be evaluated exactly $\psi_{D}(\theta) = -\nabla_\theta U(\theta)$. This allows us to estimate our density model $q_v(\theta)$ without requiring samples from the posterior distribution. Since sampling from the posterior distribution is computationally costly, we sample from the simpler and cheaper surrogate induced distribution (our variational model) instead. The corresponding expected squared distance is 
\begin{equation}\label{eq:scorevar}
\tilde{J}(v) = \frac12\int q_v(\theta)\|\psi_v(\theta)-\psi_D(\theta)\|^2\;d\theta
\end{equation}
Note that the model score function $\psi_v(\theta)=-\nabla_\theta z(\theta)$, $\tilde{J}(v)$ is exactly the expected squared distance $\tilde{D}_{SM}$ we introduced in section \ref{sec:ffvb}. In the case that our density model is not degenerate, we have a similar local consistency result to \cite{hyvarinen05} as shown by the following theorem.
\begin{theorem}\label{thm:consistency}
Assume that the data density $p_D(.)$ follows the model: $p_D(.) = q_v(.)$ for some $v^{\ast}$. Further, assume that no other parameter value gives a probability density that is equal to $q_{v^{\ast}}(.)$, and that $q_v(\theta)>0$ for all $v,\theta$. Then
\[
\tilde{J}(v) = 0 \Leftrightarrow v = v^{\ast} 
\]
\end{theorem}
As such, minimizing the expected squared distance $\tilde{D}_{SM}$ would be sufficient to estimate the model. For a proof, see Supplementary Materials.

{\bf Remark.} Note that the data score function is 
\[
\psi_D(\theta) = \sum_{n=1}^N\nabla_\theta\log p(y_n|\theta) + \nabla_\theta\log p(\theta)
\]
we may choose our surrogate function as 
\begin{equation}\label{eq:exact}
z(\theta) = \sum_{n=1}^N\log p(y_n|\theta) + \log p(\theta)
\end{equation}
then the data density $p_D$ follows the model exactly. In order to reduce computational cost, our model is usually much simpler than \eqref{eq:exact} ($s\ll N$). This allows us to explore and exploit the redundancy in the data from a function approximation perspective.

In practice, samples from the surrogate induced distribution can be collected as observations and our surrogate can be trained by minimizing the empirical version of $\tilde{J}(v)$. A regularization term could be included to improve numerical stability.

Now suppose that we have collected training data of size $t$ from the Markov chain history
\begin{equation}\label{eq:scoredata}
\mathcal{T}_s^{(t)} := \{(\theta_n,\nabla_\theta U(\theta_n))\}_{n=1}^t \in \mathbb{R}^d\times \mathbb{R}^d
\end{equation}
where $\theta_n$ is the $n$-th sample. The estimator of the output weight vector (variational parameter) can be obtained by minimizing the empirical square distance between the gradient of the surrogate and the gradient of the potential (i.e., score function) plus an additional regularization term:
\begin{equation}\label{eq:sm}
\hat{v} = \arg\min_v \frac12\sum_{n=1}^t\|\nabla_\theta z(\theta_n) - \nabla_\theta U(\theta_n)\|^2 + \frac{\lambda}{2} \|v\|^2
\end{equation}
which has an online updating formula summarized in the following proposition \ref{prop:onlinescore} (see Supplementary Materials for a detailed proof). For ease of presentation, we use the additive nodes $\sigma(\theta;\gamma_i) = \sigma(w_i\cdot\theta + d_i)$\footnote{Note that our method works for a general class of nonlinear nodes, including the radial basis functions (RBF).}
\begin{proposition}\label{prop:onlinescore}
Suppose our current estimator of the output weight vector is $v^{(t)}$ based on the current training dataset $\mathcal{T}_s^{(t)} := \{(\theta_n,\nabla_\theta U(\theta_n))\}_{n=1}^t \in \mathbb{R}^d\times \mathbb{R}^d$ using $s$ hidden neurons. Given a new training data point $(\theta_{t+1},\nabla_\theta U(\theta_{t+1}))$, the updating formula for the estimator is given by
\begin{equation}\label{eq:update}
v^{(t+1)} = v^{(t)} + W^{(t+1)}(\nabla_\theta U(\theta_{t+1}) - A_{t+1}v^{(t)})
\end{equation}
where
\begin{align*}
W^{(t+1)} &= C^{(t)} A'_{t+1}\left[I_d + A_{t+1}C^{(t)}A_{t+1}'\right]^{-1} \\
A_{t+1} &= (A_1(\theta_{t+1}),\ldots,A_s(\theta_{t+1}))
\end{align*}
with $\quad A_i(\theta_{t+1}) := \sigma'(w_i\cdot\theta_{t+1}+d_i)w_i $, and $C^{(t)}$ can be updated by {\emph{Sherman-Morrison-Woodbury}}formula:
\begin{equation}\label{eq:smw}
C^{(t+1)} = C^{(t)} - W^{(t+1)}A_{t+1}C^{(t)}
\end{equation}
\end{proposition}
The estimator and inverse matrix can be initialized as $v^{(0)} = 0,\; C^{(0)} = \frac1\lambda I_s$. The online learning can be achieved by storing the inverse matrix $C$ and performing the above updating formulas, which cost $\mathcal{O}(d^3 + ds^2)$ computation and $\mathcal{O}(s^2)$ storage independent of $t$.

\subsection{Variational HMC in Practice}
\begin{algorithm*}[!ht]
\begin{algorithmic}
\STATE{\bfseries Input:} Regularization coefficient $\lambda$, transition function $\mu_t$, number of hidden neurons $s$, starting position $\theta^{(1)}$ and HMC parameters
\STATE {\it Find the Maximum A Posterior $\theta^L$ and compute the Hessian matrix $\nabla^2_\theta U(\theta^L)$}
\STATE {\it Randomly assign the node parameters: $\{\gamma_i\}_{i=1}^s$}
\FOR{$t =1,2,\ldots,T$}
  \STATE {\it Propose $(\theta^{\ast},r^{\ast})$ with regularized surrogate induced Hamiltonian flow, using $\nabla_\theta V_t(\theta)$}
  \STATE {\it Perform Metropolis-Hasting step according to the underlying distribution $\pi_t\sim \exp(-V_t(\theta)-K(r))$} 
  \IF{New state is accepted \& $t<t_0$} 
  \STATE {\it Acquire new training data point $(\theta_{t+1},\nabla_\theta U(\theta_{t+1}))$} 
  \STATE {\it Update the output weight estimate $v^{(t+1)} \leftarrow \eqref{eq:update}$ and the inverse matrix $C^{(t+1)}\leftarrow \eqref{eq:smw}$ }
  \ELSE
  \STATE $v^{(t+1)} = v^{(t)},\; C^{(t+1)} = C^{(t)}$
  \ENDIF
  \ENDFOR
 \end{algorithmic}
\caption{Variational Hamiltonian Monte Carlo}
\label{alg:vhmc}
\end{algorithm*}

The neural network based surrogate is capable of approximating the potential energy function well when there is enough training data. However, the approximation could be poor when only few training data are available which is true in the early stage of the Markov chain simulations. To alleviate this issue, we propose to add an auxiliary regularizer which provides enough information for the sampler at the beginning and gradually diminishes as the surrogate becomes increasingly accurate. Here, we use the Laplace's approximation to the potential energy function but any other fast approximation could be used. The regularized surrogate approximation then takes the form
\[
V_t(\theta) = \mu_tz_t(\theta) + \frac12(1-\mu_t) (\theta-\theta^L)'\nabla^2_\theta U(\theta^L)(\theta-\theta^L)
\]
where $\theta^L$ is the maximum \emph{a posteriori} (MAP) estimate and $\mu_t\in[0,1]$ is a smooth monotone function monitoring the transition from the Laplace's approximation to the surrogate approximation. Refining the surrogate approximation by acquiring training data from simulating the regularized surrogate induced Hamiltonian flow, we arrive at an efficient approximate inference method: {\it Variational Hamiltonian Monte Carlo (VHMC)} (Algorithm \ref{alg:vhmc}). 

In practice, the surrogate approximation may achieve sufficient quality and an early stopping could save us from inefficient updating of the output weight vector. In fact, the stopping time $t_0$ serves as a knob to control the desired approximation quality. Before stopping, VHMC acts as a {\it free-form} variational Bayes method that keeps improving itself by collecting training data from the history of the Markov chain. After stopping, VHMC performs as a standard HMC algorithm which samples from the surrogate induced distribution. VHMC successfully combines the advantages of both variational Bayes and Hamiltonian Monte Carlo, resulting in higher computational efficiency compared to HMC and better approximation compared to VB.

\section{Experiments}\label{sec:examples}
In this section, four experiments are conducted to verify the effectiveness and efficiency of the proposed Variational HMC framework. In the first two examples, we demonstrate the performance of VHMC from a pure approximation perspective and compare it to state-of-the-art \emph{fix-form} variational Bayes methods. We then test the efficiency of VHMC on two machine learning problems with large datasets and compare our methods to SGLD \citep{Welling11}, which is one of the state-of-the-art stochastic gradient MCMC methods. Compared to standard HMC, Variational HMC introduces some additional hyper-parameters, such as the number of random bases $s$ and the transition monitor $\mu_t$, that might be hard to tune in practice. Generally speaking, we want our method to be as efficient as possible while maintaining good approximation. Therefore, we choose $s$ to be small as long as a reasonable level of accuracy can be achieved and $\mu_t$ to stay small until enough data have been collected to stabilize the training procedure. These can all be done by monitoring the acceptance rate using an initial run as in \cite{NNSHMC}. In our experiments, we found that $\mu_t = 1 - \exp(-t/n_s)$ is a useful schedule where $n_s$ can be adjusted to accommodate different transition speed according to the complexity of the model.

\subsection{A Beta-binomial Model for Overdispersion}
We first demonstrate the performance of our variational Hamiltonian Monte Carlo method on a toy example from \cite{albert09}, which considers the problem of estimating the rates of death from stomach cancer for the largest cities in Missouri. The data is available from the R package LearnBayes which consists of $20$ pairs $(n_j,y_j)$ where $n_j$ records the number of individuals that were at risk for cancer in city $j$, and $y_j$ is the number of cancer deaths that occurred in that city. The counts $y_j$ are overdispersed compared to what would be expected under a binomial model with a constant probability; therefore, \cite{albert09} assumes a beta-binomial model with mean $m$ and precision $K$:
\[
p(y_j|m,K) = {n_j \choose y_j} \frac{B(Km+y_j,K(1-m)+n_j-y_j)}{B(Km,K(1-m))}
\]
and assigns the parameters the following improper prior:
\[
p(m,K) \propto \frac1{m(1-m)}\frac{1}{(1+K)^2}
\]
The resulting posterior is extremely skewed (as shown in the bottom right corner in Figure \ref{fig:bbocontour}) and a reparameterization
\[
x_1=\mathrm{logit}(m),\;x_2=\mathrm{logit}(K)
\]
is proposed to ameliorate this issue.

\begin{figure}[!t]
\centering
\includegraphics[width=.6\textwidth]{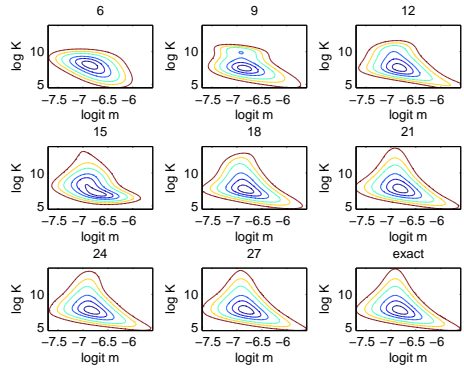}
\caption{Approximate posteriors for a varying number of hidden neurons. Exact posterior at bottom right.}
\label{fig:bbocontour}
\end{figure} 

\begin{figure}[!t]
\centering
\includegraphics[width=.6\textwidth]{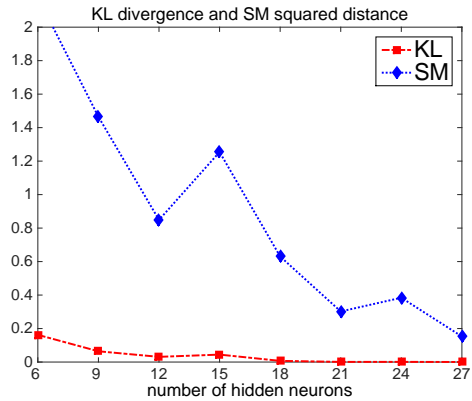}
\caption{KL-divergence and score matching squared distance between the surrogate approximation and the exact posterior density using an increasing number of hidden neurons.}
\label{fig:bboklsm}
\end{figure} 

We choose $\mu_t = 1 - \exp(-t/200)$ as our transition schedule and set up the HMC parameter to achieve around $85\%$ acceptance. We run the variational Hamiltonian Monte Carlo long enough so that we can estimate the full approximation qualify of our surrogate. To demonstrate the approximation performance in terms of the number of hidden neurons $s$, we train the neural network based surrogate using different numbers of hidden neurons and examine the resulting KL-divergence and score matching squared distance to the true posterior density. As we can see from Figures \ref{fig:bbocontour} and \ref{fig:bboklsm}, the neural network based surrogate indeed offers a high quality approximation and becomes more accurate as $s$ increases. The surrogate induced Hamiltonian flow effectively explores the parameter space and transfers information from the posterior to the surrogate.

\subsection{Bayesian Probit Regression}
Next, we demonstrate the approximation performance of our Variational HMC algorithm relative to existing variational approaches on a simple Bayesian classification problem, binary probit regression. Given $N$ observed data pairs $\{(y_n,x_n)|y_n\in\{0,1\},\;x_n\in\mathbb{R}^d\}_{n=1}^N$, the model comprised a probit likelihood function $P(y_n=1|\theta) = \Phi(\theta^Tx_n)$ and a Gaussian prior over the parameter $p(\theta) = \mathcal{N}(0,100)$, where $\Phi$ is the standard Gaussian cdf. A full covariance multivariate normal approximation is used for all variational approaches. The synthetic data we use are simulated from the model, with $N=10000$ and $d=5$. We show the performance averaged over 50 runs for all methods. We compare our algorithm to Variational Bayesian Expectation Maximization (VBEM) \citep{beal02,ormerod10}, and the best fixed-form variational approximation of \cite{salimans13} that uses Hessian and subsampling (FF-Minibatch). VBEM and FF-Minibatch are both implemented using the code provided in \cite{salimans13} with the same learning rate and minibatch size. For all variational approaches, we initialize the posterior approximation to the prior. For our Variational HMC algorithm, we choose $s=100$ random hidden units for the surrogate and set the starting point to be the origin. The number of hidden units is chosen in such a way that the surrogate is flexible enough to fit the target well and remain fast in computation. The HMC parameters are set to make the acceptance probability around $70\%$.  The target density is almost Gaussian, and a fast transition $\mu_t = 1-\exp(-t/5)$ is enough to stabilize our algorithm. Since this experiment is on synthetic data, we follow \cite{salimans13} to assess the approximation performance in terms of the root mean squared error (RMSE) between the estimate (variational mean for VB and sample mean for VHMC) and the true parameter that is used to generate the dataset. 

\begin{figure}[!t]
\centering
\centerline{\includegraphics[width=.6\textwidth]{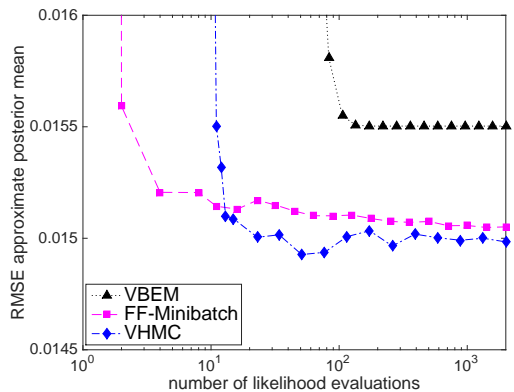}}
\caption{RMSE approximate posterior mean as a function of the number of likelihood evaluations for difference variational Bayesian approaches and our Variational HMC algorithm.}
\label{fig:rmseBPR}
\end{figure} 

Figure \ref{fig:rmseBPR} shows the performance of our Variational HMC algorithm, as well as the performance of the other two variational Bayes methods. As we can see from the graph, VHMC and the subsampling based fixed-form variational approach (FF-minibatch) achieve lower RMSE than the VBEM algorithm. That is because of the extra factorization assumptions made by VBEM when introducing the auxiliary variables \citep{ormerod10}. Even though Gaussian approximation is already sufficiently accurate  on this simple example, VHMC can still arrive at a lower RMSE due to the extra flexibility provided by the {\it free-form} neural network surrogate function.

\subsection{Bayesian Logistic Regression}
Next, we test the efficiency of our Variational HMC method as a scalable sampling method. We first apply Variational HMC to a Bayesian logistic regression model. Given the $i$-th input vector $x_i$, the corresponding output (label) $y_i=\{0,1\}$ is assumed to follow the probability $p(y_i=1|x_i,\beta) = 1/(1+\exp(-x_i^T\beta))$ and a Gaussian prior $p(\beta) = \mathcal{N}(0,100)$ is used for the model parameter $\beta$. We test our proposed algorithm on the $\mathtt{a9a}$ dataset \citep{lin08}. The original dataset, which is compiled from the UCI $\mathtt{adult}$ dataset, has 32561 observations and 123 features. We use a $50$ dimension random projection of the original features. We choose $s=2000$ hidden units for the surrogate and set a transition schedule $\mu_t = 1-\exp(-t/500)$ for our VHMC algorithm. We then compare the algorithm to HMC \citep{duane87,neal11} and to stochastic gradient Langevin dynamics (SGLD) \citep{Welling11}. For HMC and VHMC, we set the {\it leap-frog} stepsize such that the acceptance rate is around $70\%$. For SGLD we choose batch size of 500 and use a range of fixed stepsizes.

To illustrate the sampling efficiency of all methods, we follow \cite{Ahn12} to investigate the time normalized effective sample size (ESS)\footnote{Given $B$ samples, ESS = $B[1+2\sum_{k=1}^K\gamma(k)]^{-1}$, where $\gamma(k)$ is the sample autocorrelation at lag $k$ \citep{geyer92}} averaged over the $51$ parameters and compare this with the relative error after a fixed amount of computation time. The relative error of mean (REM) and relative error of covariance (REC) is defined as 
\begin{equation}\label{eq:REMC}
\mathrm{REM}_t = \frac{\sum_{i}|\overline{\beta_i^t} - \beta_i^o|}{\sum_i|\beta_i^o|},\quad \mathrm{REC}_t = \frac{\sum_{ij}|C_{ij}^t - C_{ij}^o|}{\sum_{ij}|C_{ij}^o|}
\end{equation}
where $\overline{\beta^t} = \frac1t\sum_{t'=1}^t\beta_{t'}, \; C^t = \frac1t\sum_{t'=1}^t(\beta_{t'}-\overline{\beta_t})(\beta_{t'}-\overline{\beta_t})^T$ are the sample mean and sample covariance up to time $t$ and the ground truth $\beta^o,\; C^o$ are obtained using a long run (T = 500K samples) of HMC algorithm. 

\begin{figure}[!t]
\centering
\centerline{\includegraphics[width=0.45\textwidth]{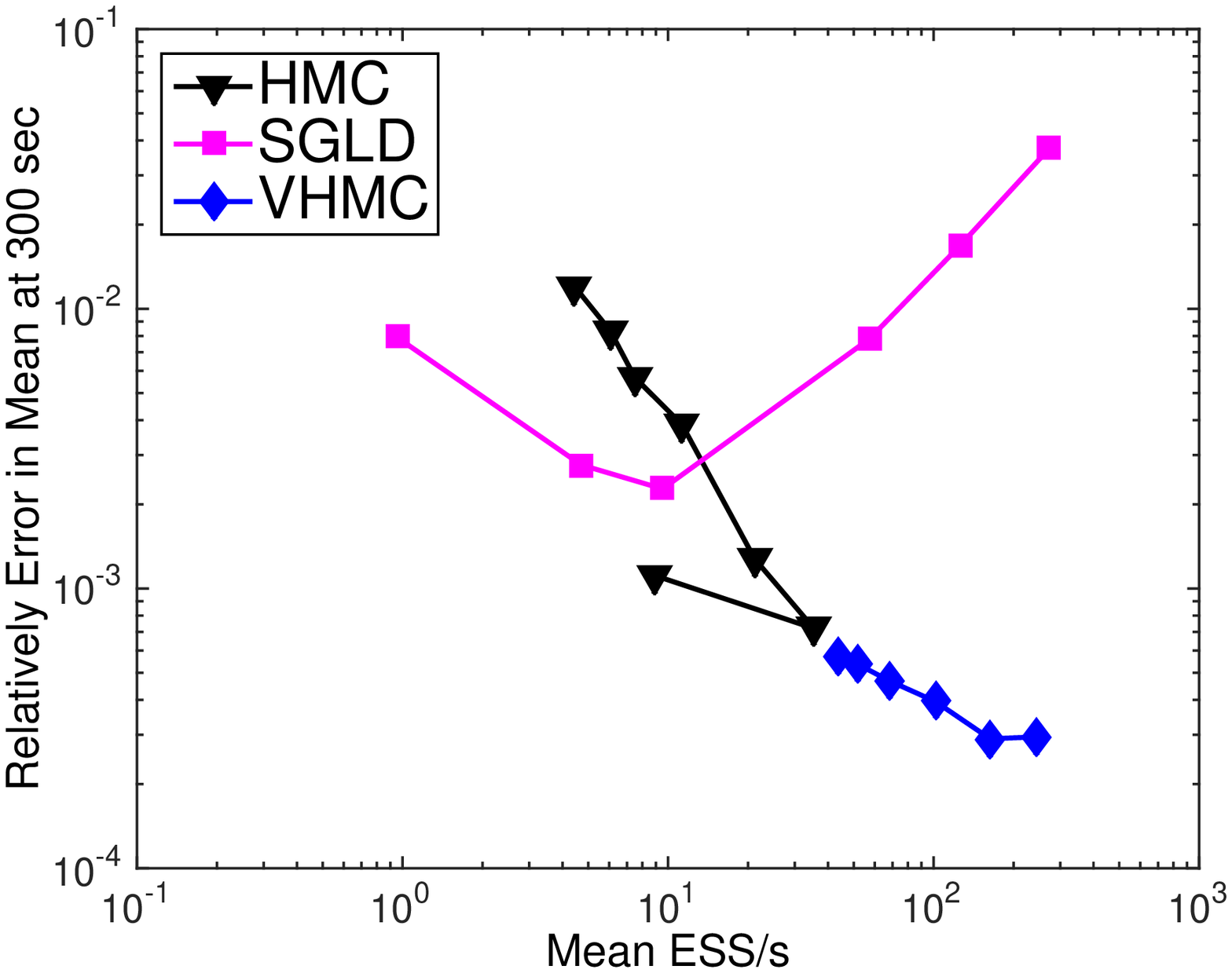}\hspace{-0pt}\includegraphics[width=0.45\textwidth]{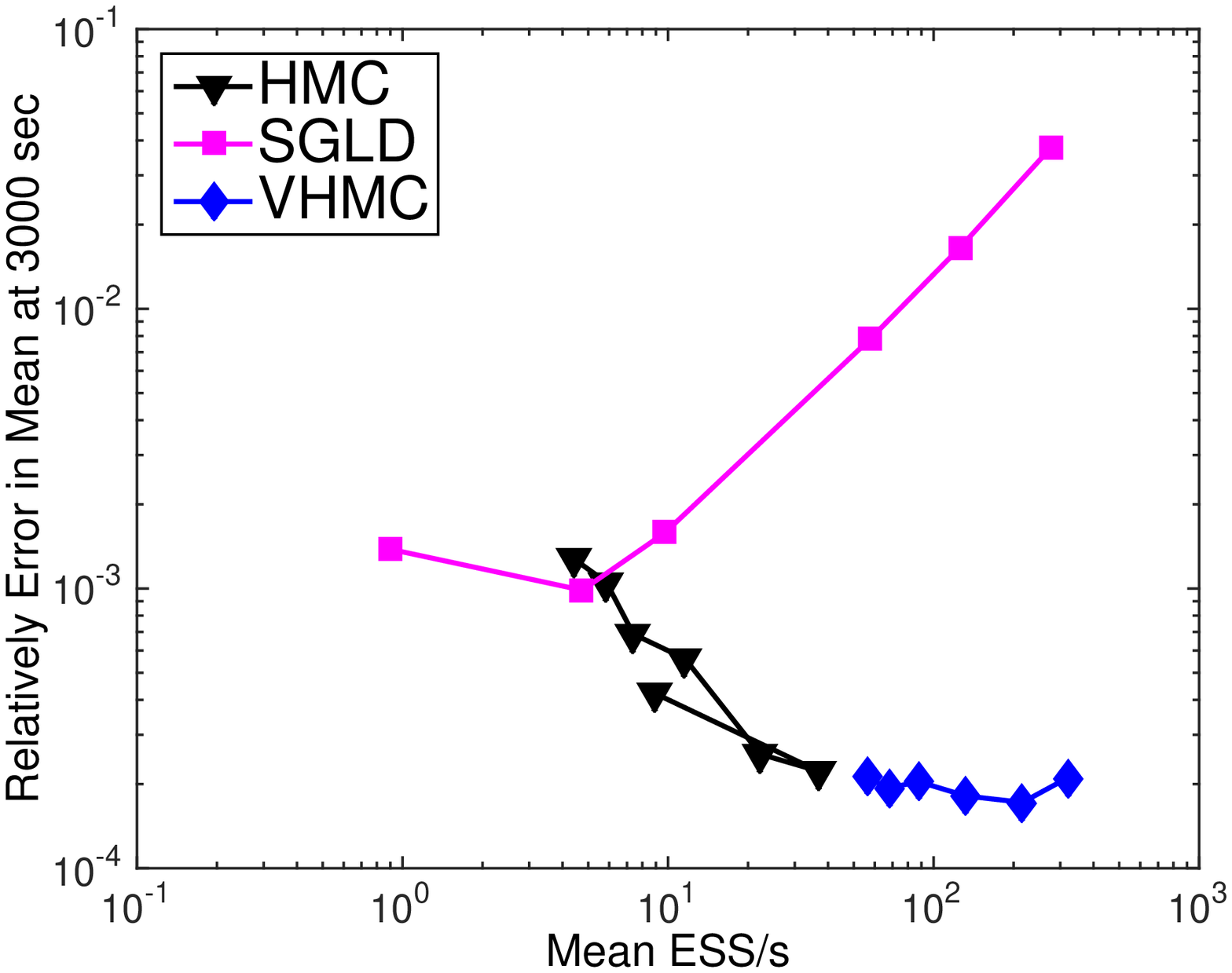}}
\centerline{\includegraphics[width=0.45\textwidth]{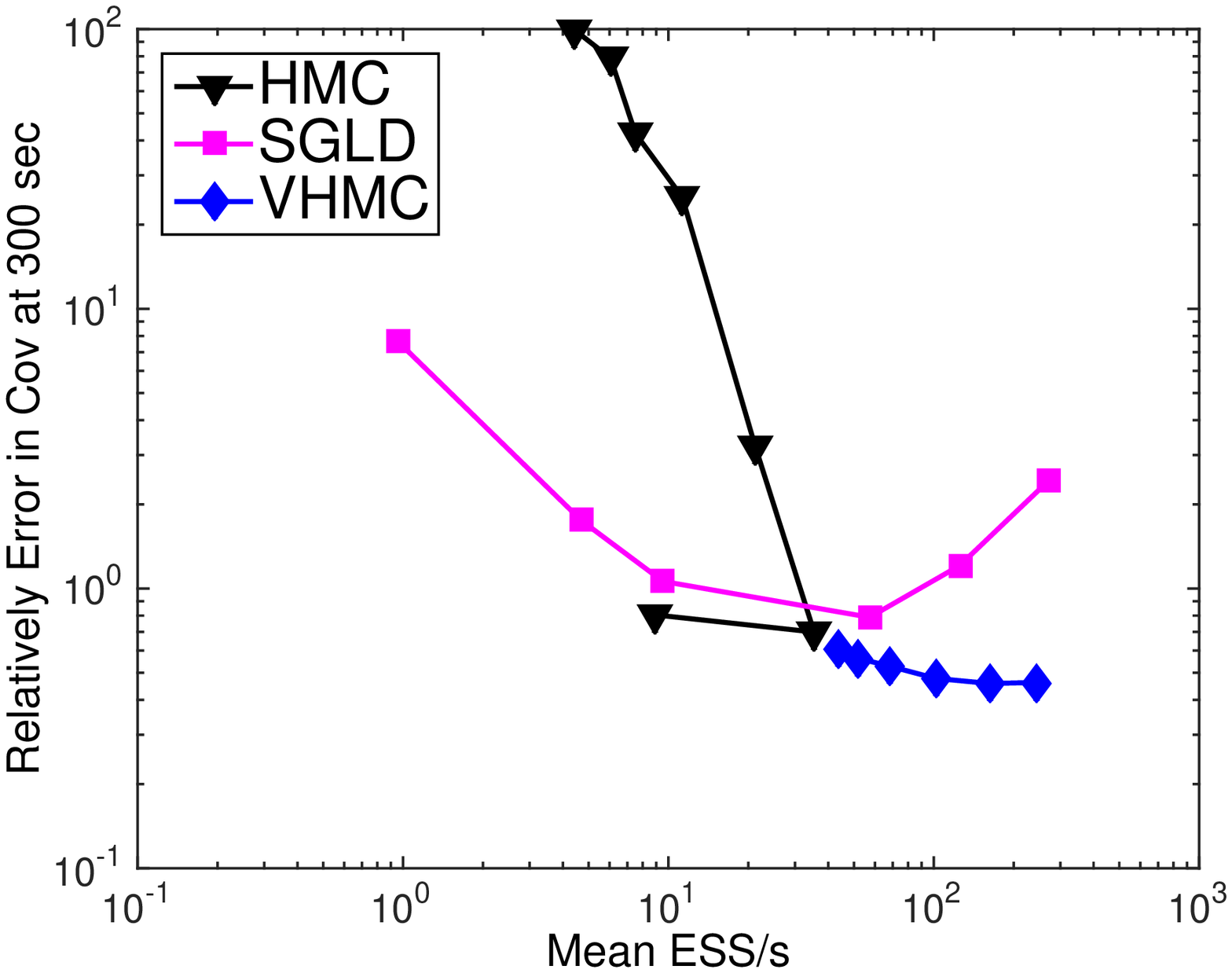}\hspace{-0pt}\includegraphics[width=0.45\textwidth]{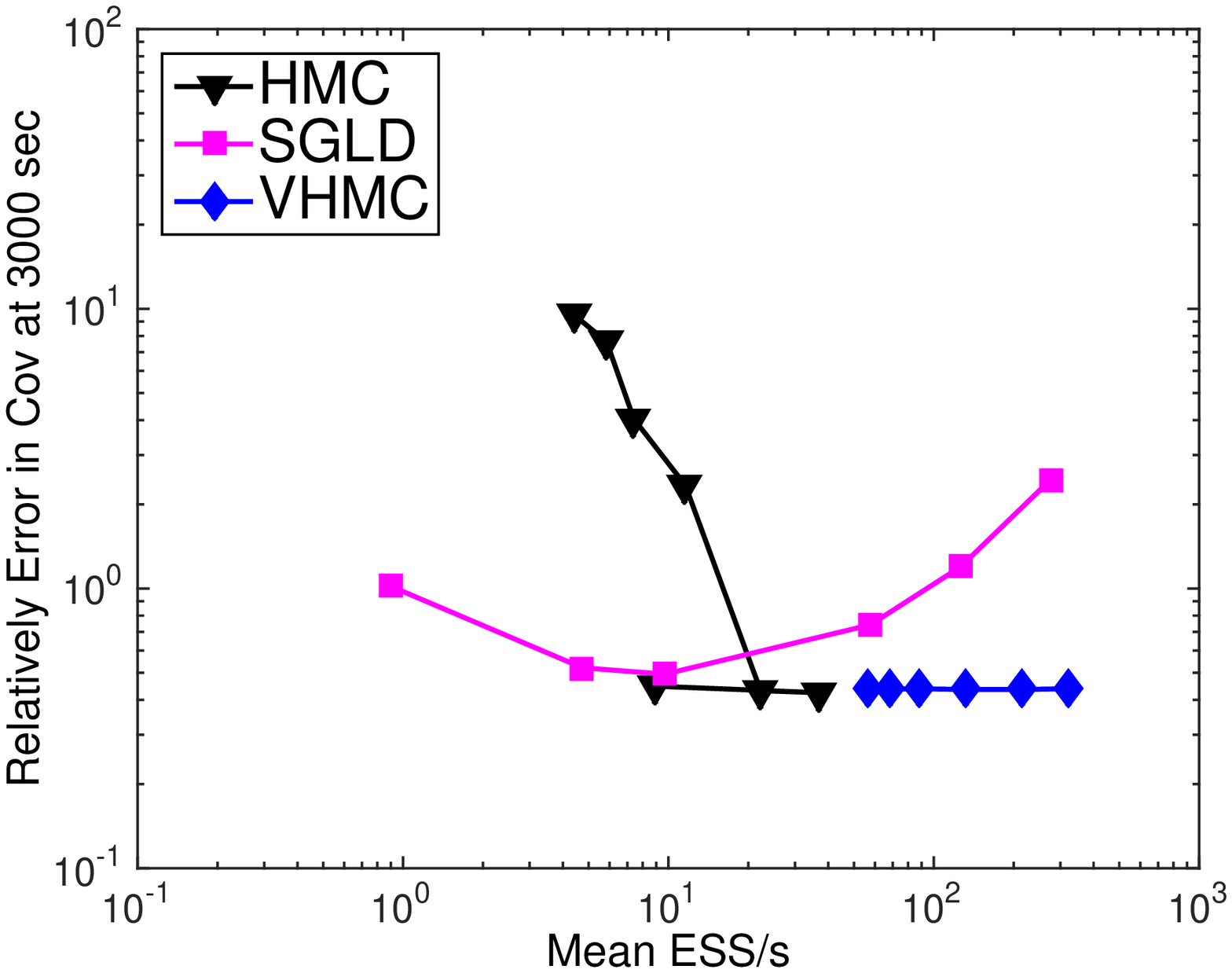}}
\caption{Final error of logistic regression at time T versus mixing rate for the mean (top) and covariance (bottom) estimates after 300 (left) and 3000 (right) seconds of computation. Each algorithm is run using different setting of parameters.}
\label{fig:remc}
\end{figure}

Figure \ref{fig:remc} shows the relative error at time T = 300, T = 3000 as a function of the time normalized mean ESS, which is a measure of the mixing rate. The results for the mean are shown on the top, and those for the covariance are on the bottom. We run each algorithm with a different setting of parameters that control the mixing rate: number of {\it leap-frog} steps $L=[50,40,30,20,10,5,1]$ for HMC and $L=[50,40,30,20,10,5]$ for VHMC, and stepsizes $\epsilon = [2e\mhyphen3,1e\mhyphen3, 5e\mhyphen4, 1e\mhyphen4,5e\mhyphen5, 1e\mhyphen5]$ for SGLD.

As we decrease the stepsize, SGLD becomes less biased in the gradient approximation, resulting in smaller relative error. However, the exploration efficiency drops at the same time and sampling variance gradually dominates the relative error. In contrast, HMC uses a fixed {\it leap-frog} stepsize and therefore maintains high exploration efficiency in parameter space. The down side is the expensive computation of the full gradient and the possible turning back of the trajectories when the number of {\it leap-frog} steps is unnecessarily large. Adopting a flexible neural network surrogate, VHMC balances the computation cost and approximation quality much better than subsampling and achieves lower relative error with high mixing rates. 

\subsection{Independent Component Analysis}
Finally, we apply our method to sample from the posterior distribution of the unmixing matrix in Independent Component Analysis (ICA) \citep{hyvarinen00}. Given $N$ $d$-dimensional observations $X=\{x_n\in\mathbb{R}^d\}_{n=1}^N$, we model the data as 
\begin{equation}\label{eq:ica}
p(x|W) = |\det(W)|\prod_{i=1}^dp_i(w_i^Tx)
\end{equation}
where $w_i$ is the $i$-th row of $W$ and $p_i$ is supposed to capture the true density of the $i$-th independent component. Following \cite{Welling11}, we use a Gaussian prior over the unmixing matrix $p(w_{ij}) = \mathcal{N}(0,\sigma)$ and choose $p_i(y_i) = [4\cosh(\frac12y_i)]^{-1}$ with $y_i = w_j^Tx$. We evaluate our method using the MEG dataset \citep{vigario97}, which has 122 channels and 17730 observations. We extract the first 5 channels for our experiment which leads to samples with $25$ dimensions. We then compare our algorithm to standard HMC and SGLD \citep{Welling11}. For SGLD, we use a natural gradient \citep{amari96} which has been found to improve the efficiency of gradient descent significantly. We set $\sigma = 100$ for the Gaussian prior. For HMC and Variational HMC, we set the {\it leap-frog} stepsize to keep the acceptance ratio around $70\%$ and set $L=40$ to allow an efficient exploration in parameter space. For SGLD, we choose batch size of $500$ and use stepsizes from a polynomial annealing schedule $a(b+t)^{-\delta}$, with $a = 5\times 10^{-3}, b = 10^{-4}$ and $\delta = 0.5$. (This setting reduces the stepsize from $5\times 10^{-5}$ to $1\times 10^{-6}$ during 1e+7 iterations). We choose $s=1000$ hidden units and set the transition schedule $\mu_t = 1-\exp(-t/2000)$ for our Variational HMC algorithm. Note that the likelihood \eqref{eq:ica} is row-permutation invariant. To measure the convergence of all samplers on this ICA problem, we therefore use the Amari distance \citep{amari96} $d_A(\overline{W},W_0)$, where $\overline{W}$ is the sample average and $W_0$ is the true unmixing matrix estimated using a long run (T = 100K samples) of standard HMC algorithm.

\begin{figure}[!t]
\centering
\centerline{\includegraphics[width=.6\textwidth]{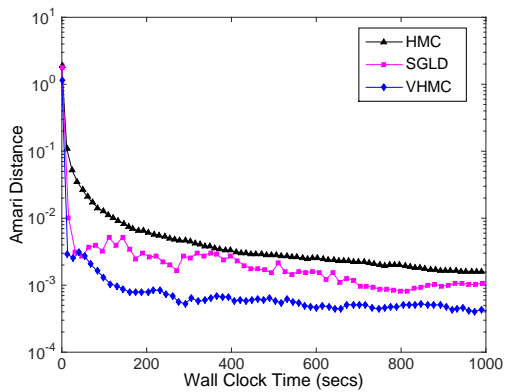}}
\caption{Convergence of Amari distance on the MEG data for HMC, SGLD and our Variational HMC algorithm.}
\label{fig:daica}
\end{figure} 

The Amari distance as a function of runtime is reported for each of these methods in Figure \ref{fig:daica}. From the graph, we can see that SGLD converges faster than standard HMC. The noise introduced by subsampling is compensated by the fast exploration in parameter space which allows for an early arrival at the high probability domain for SGLD. However, when the variance dominates the bias, the convergence speed slows down since SGLD requires annealing (or small) stepsize that inevitably leads to low exploration efficiency. By maintaining efficient exploration in parameter space (same stepsize as HMC) while reducing the computation in simulating the Hamiltonian flow, VHMC outperforms SGLD, arriving at a lower Amari distance much more rapidly.

\section{Conclusion}\label{sec:conc}
We have presented a novel approach, Variational Hamiltonian Monte Carlo, for approximate Bayesian inference. Our approach builds on the framework of HMC, but uses a flexible and efficient neural network surrogate function to approximate the expensive full gradient. The surrogate keeps refining its approximation by collecting training data while the sampler is exploring the parameter space. This way, our algorithm can be viewed as a {\it free-form} variational approach. Unlike subsampling-based MCMC methods, VHMC maintains the relatively high exploration efficiency of its MCMC counterparts while reducing the computation cost. Compared to {\it fixed-form} variational approximation, VHMC is more flexible and thus can approximate more general target distribution better.

As the complexity of the model increases (e.g., high dimensional covariates), the computational cost of random network surrogates may increase since it usually requires more random bases to provide adequately accurate approximation in high dimensions. However, the free-style surrogate construction in VHMC allows us to reduce the computation cost by adapting to the model if some special structures exist. For example, one can build surrogates based on graphical models to exploit the conditional independence structures in the model. Alternatively, one can also train the full neural network with sparsity constraints imposed to the input weights. Both approaches could be interesting to explore in future work.

\subsection*{Acknowledgement}
This work is supported by NIH grant R01AI107034 and NSF grants DMS-1418422 and DMS-1622490.

\bibliographystyle{ba}
\bibliography{example_paper}

\end{document}